\documentclass[a4paper]{article}
\usepackage[affil-it]{authblk}
\usepackage[utf8]{inputenc}
\usepackage{graphicx}
\usepackage{url}
\usepackage{hyperref}
\usepackage{amsmath}
\usepackage{oz}
\usepackage{xcolor}

\newtheorem{definition}{Definition}
\newtheorem{proposition}[definition]{Proposition}

\title{Millions of 5-State $n^3$ Sequence Generators via  Local Mappings}

\author{Tien Thao Nguyen}
\author{Luidnel Maignan}
\affil{Univ Paris Est Creteil, LACL, F-94010 Creteil, France}
\date{}

\newcommand{\stateSet}{{\rm \Sigma}}
\newcommand{\oState}{{\tt \star}}

\newcommand{\qState}{{\tt Q}}

\newcommand{\bState}{{\tt B}}

\newcommand{\sState}{{\tt S}}

\newcommand{\initConfigSet}{{\rm I}}

\newcommand{\initConfigRTSG}{\overrightarrow{\infty}}
\newcommand{\lConfig}[3]{(#1,#2,#3)}

\newcommand{\transTab}{{\rm \delta}}
\newcommand{\fDiagram}{{\rm D}}
\newcommand{\diagram}[4]{\fDiagram_{#1}({#2})({#3},{#4})}

\newcommand\family[1]{\mathrm{D}_{#1}}

\newcommand{\ie}{\emph{i.e.}~}

\newcommand{\lmap}{{\ell}}

\begin{document}

\maketitle

\begin{abstract}
In this paper, we come back on the notion of local simulation allowing to transform a cellular automaton into a closely related one with different local encoding of information.
In a previous paper, we applied it to the Firing Squad Synchronization Problem.
In this paper, we show that the approach is not tied to this problem by applying it to the class of Real-Time Sequence Generation problems.
We improve in particular on the generation of $n^3$ sequence by using local mappings to obtain millions of 5-state solution, one of them using 58 transitions.
It is based on the solution of Kamikawa and Umeo that uses 6 states and 74 transitions.
Then, we explain in which sense even bigger classes of problems can be considered.
\end{abstract}

\section{Introduction}
\label{sec:intro}

\subsection{Local Mappings to Explore the Cellular Solution Space}

This paper is about the formal concepts of \emph{local mapping} and \emph{local simulation} which found their firsts applications in the study of the so-called Firing Squad Synchronization Problem (FSSP) proposed by John Myhill in 1957.
In the latter, the goal is to find a single cellular automaton such that any one-dimensional horizontal array of an arbitrary number of cells \emph{synchronizes}, \ie such that a special state is set for all the cells at the same time.
As explained in~\cite{DBLP:conf/automata/NguyenM20}, there was a race to obtain solutions with as few states as possible, leading in 1987 to a situation where it was established, for minimal-time time solutions, that there were no 4-state solutions, and only one solution had 6-state using a unique strategy.
In 2018, a surprise came when 718 solutions were found using massive computing power, but a bigger surprise came in 2020 when many millions of solutions where discover using ordinary computer power.
It is still not known whether this is a 5 state solution, but the concept used to generate the millions of solutions, local mappings and local simulation still have many things to tell, a story that is not tied to FSSP, but that may lead to some ideas about the famous 5-state FSSP question.
More information can be found in~\cite{DBLP:journals/jca/NguyenM20,DBLP:conf/automata/NguyenM20}.

\subsection{Real-Time Sequence Generation Problems}

To show explicitly in which sense the approach can be applied to other problems, let us focus here on the so-called Real-Time Sequence Generation problems, or RTSG problems for short.
In the latter, given a fixed sequence $S \subseteq \mathbb{N}$, the goal is to find a cellular automaton running on an one-dimensional horizontal array of cells such that the leftmost cell is in a special state exactly when the number of transition $t$ from the beginning belongs to $S$.
A formal description is given later.
In the following, we write $f(n)$ to mean $S = \{ f(n) \mid n \ge 1\}$.

The study of such problems began in 1965 for the sequence of prime numbers, with a description of a cellular automaton algorithm by Fischer~\cite{DBLP:journals/jacm/Fischer65}.
In 1998, Korec~\cite{DBLP:conf/mcu/Korec98} proposed a 9-state solution.
Other sequences where considered in 2007 by Kamikawa and Umeo~\cite{4421122} who gave some different algorithms for the sequences $2^n$, $n^2$, and $3^n$ using one-bit inter-cell-communication cellular automata.
In 2012, Kamikawa and Umeo~\cite{2012JCMSI...5..191K} described the sequence generation powers of CAs having a small number of states, focusing on the CAs with one (only one sequence $n$ of all positive natural numbers), two, and three internal states, respectively. The authors enumerate all of the sequences generated by two-state CAs (linear sequences: $2n, 4n, 3n-1, n, 3n-2, 2n-1, n+1$; non-regular sequences: $2^{n+1}-2, 2^n - 1$) and present several non-regular sequences like $2^n, n^2, 3^n$ that can be generated in real-time by three-state CAs, but not generated by any two-states CA.
In 2016~\cite{DBLP:journals/alr/KamikawaU16}, they gave a construction for the Fibonacci sequence using five-state, followed in 2019~\cite{DBLP:journals/ijnc/KamikawaU19} by two solutions of 8 states and 6 states for the sequence $n^3$.
In these studies, much attention has been paid to the developments of real-time generation algorithms and their small-state implementations on CAs for specific non-regular sequences.
Other complexities are also studied such as the space, communication or state-change complexities.

Here we consider the sequence $n^3$ and provide millions of 5 states solutions using local mappings to improve on the previous 6-state solution.
This reduction from 6-state to 5-state is reminiscent of the FSSP situation.
For the particular $n^3$ sequence, this leaves open the question of the existence of solutions with 3 or 4 states.
We did not try to improve on any other work listed above, the goal here being simply to illustrate the generality of the approach.

\subsection{Organization of the Content}

In Section~\ref{sec:background}, we begin by defining formally cellular automata, local mappings, local simulations, RTSG solutions and related objects in a suitable way for this study.
In Section~\ref{sec:rtsg}, we explain how local mappings can be used firstly to obtain a first 5-state solution to the $n^3$-RTSG problem, and then to generate millions of other solutions.
These other solutions are essentially the same, but differ in the way the local information is encoded, leading to different numbers of transitions for example.
This is in direct comparison with compiler optimization where a program is optimized but stays essentially the same.
We finish this section by making more precise the generality of the approach.
We conclude in Section~\ref{conclusion} with a discussion of some additional aspects of this investigation, in particular with the relation with some topological concepts.

\section{Preliminaries}
\label{sec:background}

We summarize here the formal definitions of cellular automata, local mappings solutions as defined in~\cite{DBLP:conf/automata/NguyenM20}.
More detailed explanations can be found in~\cite{DBLP:conf/automata/NguyenM20,DBLP:journals/jca/NguyenM20}.
We then define RTSG solutions in a formally relevant way for this framework.
As for the other papers, the formal setting is presented for one dimensional cellular automata with usual neighborhood $\{-1,0,+1\}$, but is easily extended to any (non necessarily commutative) group and any (none necessarily fixed) neighborhood.

\subsection{Cellular Automata, Local Mappings, and Local Simulations}

The purpose of these following definitions is to describe cellular automata with partial transition table first directly and then in terms of their deterministic family of space-time diagrams.
With this more explicit representation, the concepts of local mapping and local simulations are more easily understood.
Non necessarily deterministic family of space-time diagrams also plays a role in this story.

\begin{definition}\label{def:ca}\label{def:std}
    A \emph{cellular automaton} $\alpha$ consists of a finite set of \emph{states} $\stateSet_\alpha$, a set of \emph{initial configurations} $\initConfigSet_\alpha \subseteq {\stateSet_\alpha}^\mathbb{Z}$ and a partial function $\transTab_\alpha : {\stateSet_\alpha}^3 \pfun \stateSet_\alpha$ called the \emph{local transition function} or \emph{local transition table}.
    The elements of ${\stateSet_\alpha}^\mathbb{Z}$ are called \emph{(global) configurations} and those of ${\stateSet_\alpha}^3$ are called \emph{local configurations}.
    For any $c \in \initConfigSet_\alpha$, its \emph{space-time diagram} $\fDiagram_\alpha(c) : \mathbb{N} \times \mathbb{Z} \to \stateSet_\alpha$ is defined as:
    \begin{equation*}
        \diagram{\alpha}{c}{t}{p} =
        \begin{cases}
            c(p) & \text{ if } t = 0,\\
            \transTab_\alpha(c_{-1}, c_0, c_1) & \text{ if } t > 0 \text{ with } c_i = \diagram{\alpha}{c}{t-1}{p+i}.
        \end{cases}
    \end{equation*}
    The partial function $\transTab_\alpha$ is required to be such that all space-time diagrams are totally defined.
    When $\fDiagram_\alpha(c)(t,p) = s$, we say that, for the cellular automaton $\alpha$ and initial configuration $c$, the cell at position $p$ has state $s$ at time $t$.
\end{definition}

\begin{definition}\label{family-std}
    A \emph{family of space-time diagrams $D$} consists of a set of states $\stateSet_D$ and an arbitrary set $D \subseteq {\stateSet_D}^{\mathbb{N} \times \mathbb{Z}}$ of space-time diagrams.
    The \emph{local transition relation $\delta_D \subseteq {\stateSet_D}^3 \times \stateSet_D$ of $D$} is defined as:
    $$((c^0_{-1}, c^0_0, c^0_1), c^1_0) \in \delta_D :\iff \exists(d,t,p) \in D \times \mathbb{N} \times \mathbb{Z} \text{ s.t. } c^j_i = d(t+j,p+i).$$
    We call $D$ a \emph{deterministic family} if its local transition relation is functional.
\end{definition}

%The family $\{\,\fDiagram_\alpha(c) \mid c \in \initConfigSet_\alpha \,\}$ of space-time diagrams generated by a cellular automaton $\alpha$ is clearly deterministic.
%Also, each deterministic family $F$ is obtained in this way from a unique CA $\mathrm{\gamma}_F$, \ie $\{\,\fDiagram_{\mathrm{\gamma}_F}(c) \mid c \in \initConfigSet_{\mathrm{\gamma}_F} \,\} = F$ with each $c \in \initConfigSet_{\gamma_F}$ corresponding to a unique $d \in F$ with $d(0,\--) = c$ and $\fDiagram_{\mathrm{\gamma}_F}(c) = d$.

\begin{definition}\label{def:family-to-ca}
    Given a deterministic family $D$, its \emph{associated cellular automaton} ${\rm \Gamma}_D$ is defined as having the set of states $\stateSet_{{\rm \Gamma}_D} = \stateSet_D$, the set of initial configurations $\initConfigSet_{{\rm \Gamma}_D} = \{ d(0,\--) \in {\stateSet_\alpha}^{\mathbb{Z}}\mid d \in D \}$, and the local transition function $\transTab_{{\rm \Gamma}_D} = \delta_{D}$.
\end{definition}
\begin{definition}
    Given a cellular automaton $\alpha$, its \emph{associated family of space-time diagrams} (abusively denoted) $\family{\alpha}$ is defined as having the set of states $\stateSet_{\family{\alpha}} = \stateSet_\alpha$, and the set of space-time diagrams $\{\,\fDiagram_\alpha(c) \mid c \in \initConfigSet_\alpha \,\}$ and is clearly deterministic.
\end{definition}

\begin{definition}
    \label{def:local-mapping}
    \label{def:generated-family}
    \label{def:local-simulation}
    A \emph{local mapping $\lmap$ from a CA $\alpha$ to a finite set $X$} consists of two functions $\lmap_{\mathtt{z}}: \{ d(0,p) \mid (d, p) \in \family{\alpha} \times \mathbb{Z} \} \to X$ and $\lmap_{\mathtt{s}}: \dom(\transTab_\alpha) \to X$.
    We define its \emph{associated family of diagrams} $\Phi_\lmap = \{ \lmap(d) \mid d \in \family{\alpha} \}$ where:
    \begin{equation*}
        \lmap(d)(t,p) =
        \begin{cases}
            \lmap_\mathtt{z}(d(0,p)) & \text{ if } t = 0,\\
            \lmap_\mathtt{s}(d(t-1,p-1), d(t-1,p), d(t-1,p+1)) & \text{ if } t > 0.
        \end{cases}
    \end{equation*}
    If $\Phi_\lmap$ is deterministic, we say that $\lmap$ is a \emph{local simulation from CA $\alpha$ to CA ${\rm \Gamma}_{\Phi_\lmap}$}.
    %Equivalently, a \emph{local simulation $h$ from a CA $\alpha$ to a CA $\beta$} is a local mapping from $\alpha$ to the set $\stateSet_\beta$ such that $\{ \lmap_{\mathtt{z}}(c) \mid c \in \initConfigSet_\alpha \} = \initConfigSet_\beta$ and for all $(c, t, p) \in \initConfigSet_\alpha \times \mathbb{N} \times \mathbb{Z}$, we have $\lmap_{\mathtt{s}}(c_{-1}, c_0, c_1) = l'_0$ with $c_i = \diagram{\alpha}{c}{t}{p + i}$ and $l'_0 = \diagram{\beta}{\lmap_{\mathtt{z}}(c)}{t + 1}{p}$.
\end{definition}

\subsection{Real-Time Sequence Generators}
\label{rtsp}

\begin{definition}\label{def:rtsg-candidate}
    A cellular automaton is \emph{RTSG-candidate} if there are four special states $\oState_\alpha, \bState_\alpha, \qState_\alpha, \sState_\alpha \in \stateSet_\alpha$, if $\initConfigSet_\alpha = \{ \initConfigRTSG_\alpha \}$ with $\initConfigRTSG_\alpha$ being the \emph{RTSG initial configuration} infinite to the right, i.e. $\initConfigRTSG_\alpha(p) = \oState_\alpha$, $\bState_\alpha$ and $\qState_\alpha$ if $p$ is respectively $p \le 0$, $p = 1$ and $p \ge 2$. 
    Moreover, $\oState_\alpha$ must be the \emph{outside state}, \emph{i.e.} for any $(c_{-1},c_0,c_1) \in \dom(\transTab_\alpha)$, we must have $\transTab(c_{-1},c_0,c_1) = \oState_\alpha$ if and only if $c_0 = \oState_\alpha$.
    Also, $\qState_\alpha$ must be a \emph{quiescent state} so $\transTab_\alpha(\qState_\alpha,\qState_\alpha,\qState_\alpha) = \transTab_\alpha(\oState_\alpha, \qState_\alpha,\qState_\alpha) = \qState_\alpha$.
\end{definition}

The $\oState_\alpha$ state is not really counted as a state since it represents cells that should be considered as non-existing.
Therefore, an RTSP-candidate cellular automaton $\alpha$ will be said to have $s$ states when $|\stateSet_\alpha \setminus \{\oState_\alpha\}|\; = s$, and $m$ transitions when $|\dom(\transTab_\alpha) \setminus\stateSet_\alpha\times\{\oState_\alpha\}\times\stateSet_\alpha| = m$.

\begin{definition}\label{rtsg_solution}
Given a sequence $S \subseteq \mathbb{N}$, a RTSG-candidate cellular automaton $\alpha$ is a \emph{$S$-RTSG solution} if for any time $t$, $\fDiagram_\alpha(\initConfigRTSG_\alpha)(t,0) = \sState_\alpha$ if and only if $t \in S$.
\end{definition}

\begin{figure}[t]
    \centering
    \includegraphics[width=.66\linewidth]{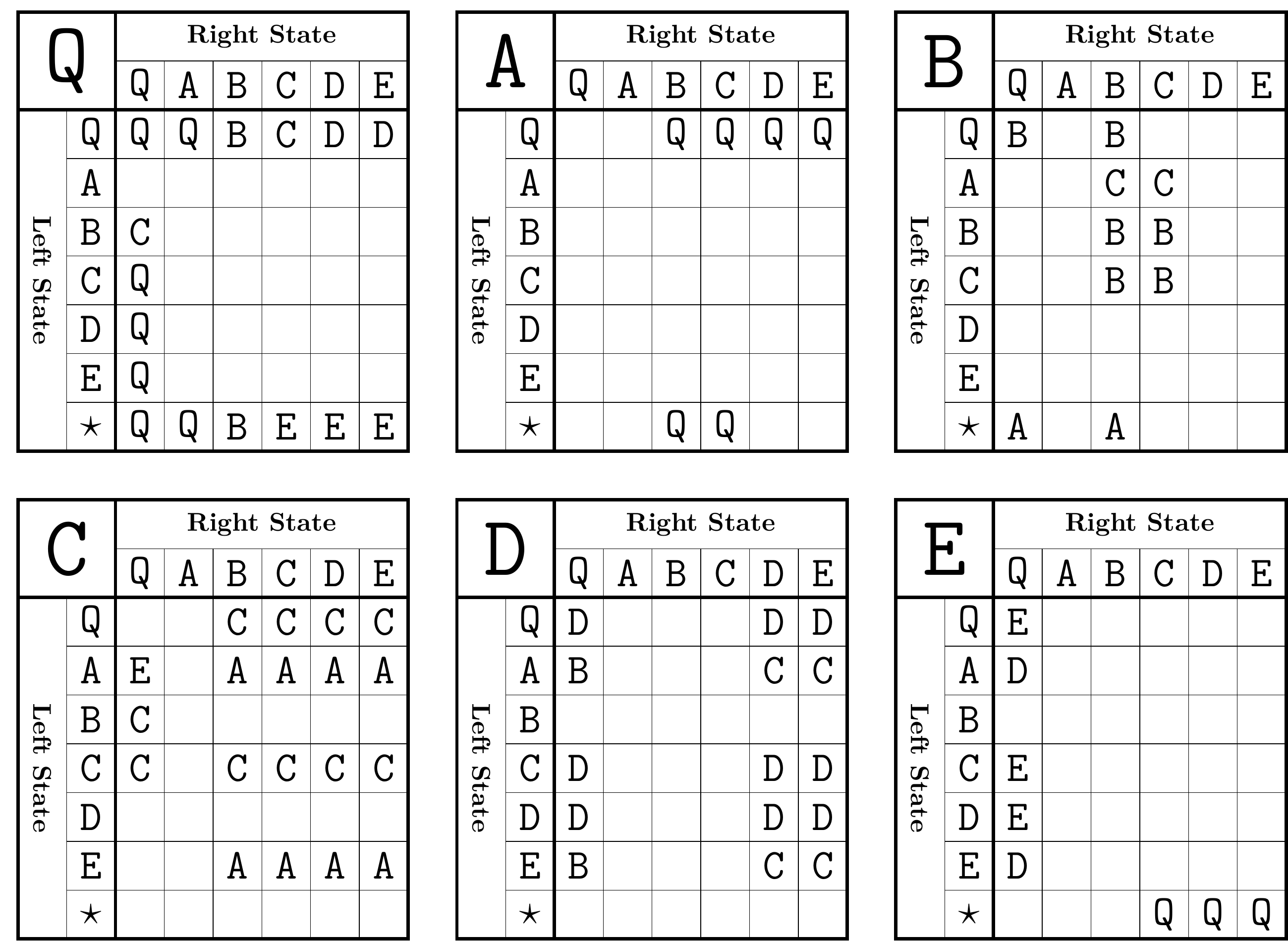}
    \caption{Transition table of Kamikawa and Umeo's 6-state solution using 74 transitions.}
    \label{fig:6state}
\end{figure}
\begin{proposition}\label{prop:umeo-6-state}
   There is a $n^3$-RTSG solution using 6 states and 74 transitions.
\end{proposition}
\begin{proof}
    In Figure~\ref{fig:6state} is the solution of Kamikawa and Umeo, reproduced with the same format as their paper to ease comparison.
    Also, the local transition function $\transTab$ contains the above entries and additional obvious entries for the outside state $\star$.
    The proof of correction can be found in~\cite{DBLP:journals/ijnc/KamikawaU19}.
\end{proof}
The space-time diagram of this solution is depicted up in the two left columns of Figure~\ref{fig:rtsg_dgm}, where the cell at position $0$ has the state $A$ at time $1$, $8$, $27$, and $64$ as expected.
In~\cite{DBLP:journals/ijnc/KamikawaU19}, note that table of $D$ wrongly has column $C$ filled with the content of column $E$.
This mistake is easy to catch by examining the proofs and space-time diagrams of the paper.

\section{Exploring RTSG Solutions and More via Local Mappings}
\label{sec:rtsg}

Let us now describe how to apply the same techniques used for the FSSP in~\cite{DBLP:conf/automata/NguyenM20} in order to first obtain a first optimization from 6-state to 5-state, and then use the exploration algorithm to generate millions of other 5-state solutions.
The first step is to study those local mappings which \emph{complies} with RTSG problems.

\subsection{Compliant Local Mappings}

Given two cellular automata $\alpha$ and $\beta$, a local mapping between them associates to each triplet found in a space-time  diagram $d$ at position $p$ and time $t$ of $\alpha$ to the state found in the associated diagram $d'$ at position $p$ and time $t+1$.
When both of these cellular automata are RTSG solutions, this implies the following properties on the local mapping.

\begin{definition}\label{rtsg_compliant}
   A local mapping $\lmap$ from an RTSG solution $\alpha$ to the states $\stateSet_\beta$ of an RTSG-candidate CA $\beta$ is said to be \emph{RTSG-compliant} if it is such that
   (0) $\lmap_\mathtt{z}$ maps $\oState_\alpha$, $\bState_\alpha$, and $\qState_\alpha$ respectively to $\oState_\beta$, $\bState_\beta$, and $\qState_\beta$,
   (1) $\lmap_\mathtt{s}\lConfig{c_{-1}}{c_0}{c_1} = \oState_\beta$ if and only if $\delta_\alpha\lConfig{c_{-1}}{c_0}{c_1} = \oState_\alpha$ (meaning simply $c_0 = \oState_\alpha$),
   (2) $\lmap_\mathtt{s}\lConfig{\oState_\alpha}{c_0}{c_1} = \sState_\beta$ if and only if $\delta_\alpha\lConfig{\oState_\alpha}{c_0}{c_1} = \sState_\alpha$, and
   (3) $\lmap_\mathtt{s}\lConfig{\qState_\alpha}{\qState_\alpha}{\qState_\alpha} = \lmap_\mathtt{s}\lConfig{\oState_\alpha}{\qState_\alpha}{\qState_\alpha} = \qState_\beta$.
\end{definition}

\begin{proposition}\label{prop:conservation}
   Given a sequence $S \subseteq \mathbb{N}$, let $\alpha$ be an $S$-RTSG solution CA, $\beta$ an $S$-RTSG-candidate CA and $\lmap$ a local simulation from $\alpha$ to $\beta$.
   $\beta$ is an $S$-RTSG solution if and only if $\lmap$ is RTSG-compliant.
\end{proposition}
\begin{proof}
    To see this, consider the diagram $d \in \fDiagram_\alpha$ of the solution $\alpha$.
    The special RTSG states appear at specific places and $\lmap$ ensures or witnesses, depending on the direction of the implication considered, that these special states/places are conserved in $\lmap(d) \in \fDiagram_\beta$, (Definition~\ref{def:local-mapping}).
    Indeed, condition (0) is just about the initial configuration, condition (1) is about the conservation of the outside state, condition (2) is about the conservation of the special generation state for the leftmost cell only and condition (3) about the conservation of the quiescent state behaviour.
    These conditions are sufficient to ensure and $\beta$ is a solution, and clearly necessary since they perfectly match Definitions~\ref{def:rtsg-candidate} and~\ref{rtsg_solution} of the problem.
\end{proof}

Note that once $\alpha$ fixed, $\beta$ can be reconstructed from $\lmap$, and $\lmap$ from $\beta$.
So the local mapping $\lmap$ is just another representation of the RTSG-candidate $\beta$ that it generates (see Definition~\ref{def:local-simulation}), but it is much easier to check the compliance of the local mapping than the correction of CA $\beta$ as an RTSG solution, and this is the key property than justifies this particular application of local mappings.

\subsection{A Hand-Crafted Local Simulation}

\begin{figure}[t]
\includegraphics[width=.4\textwidth]{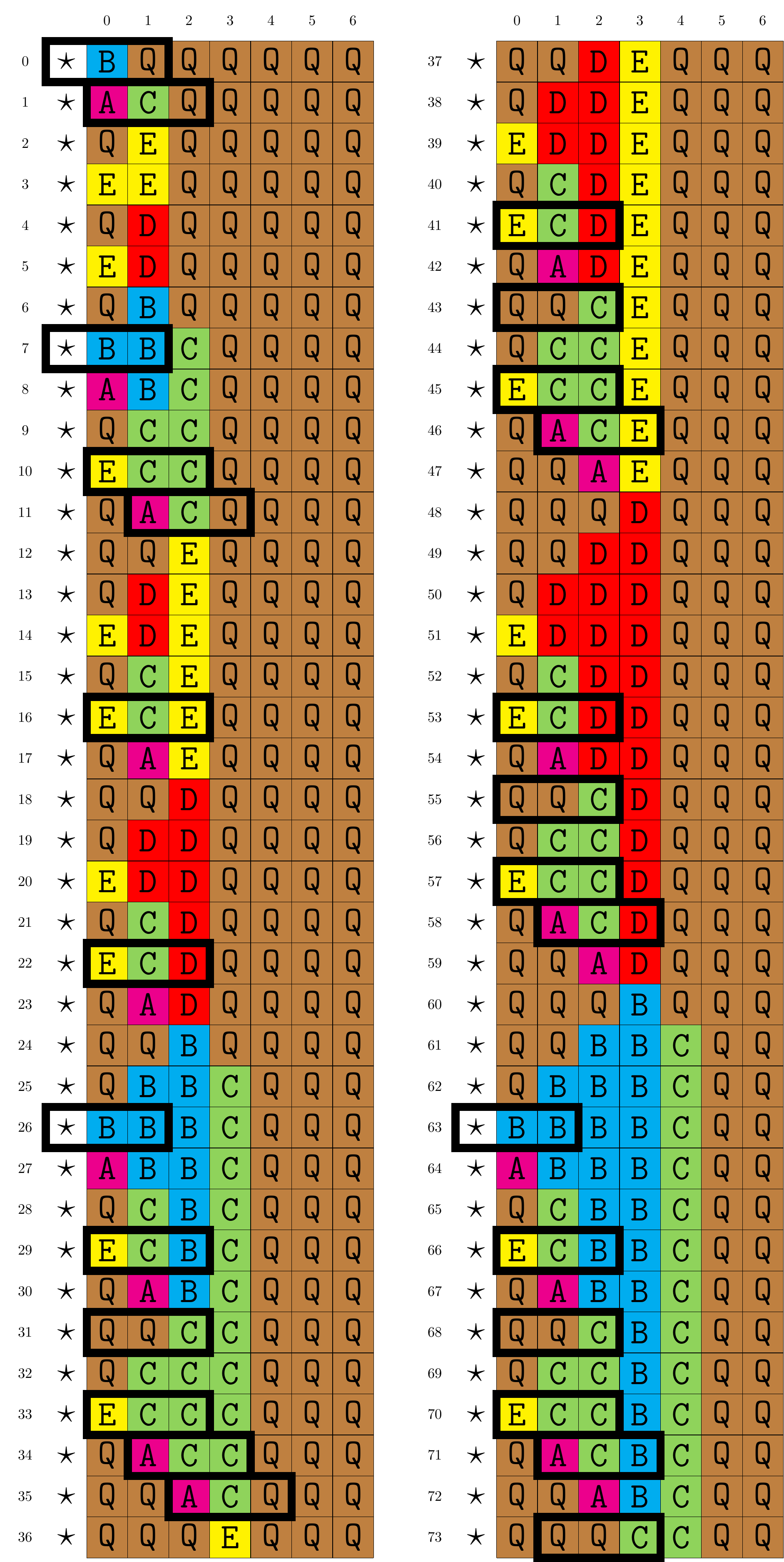}
\hspace{.2em}
\raisebox{4em}{\includegraphics[width=.15\textwidth]{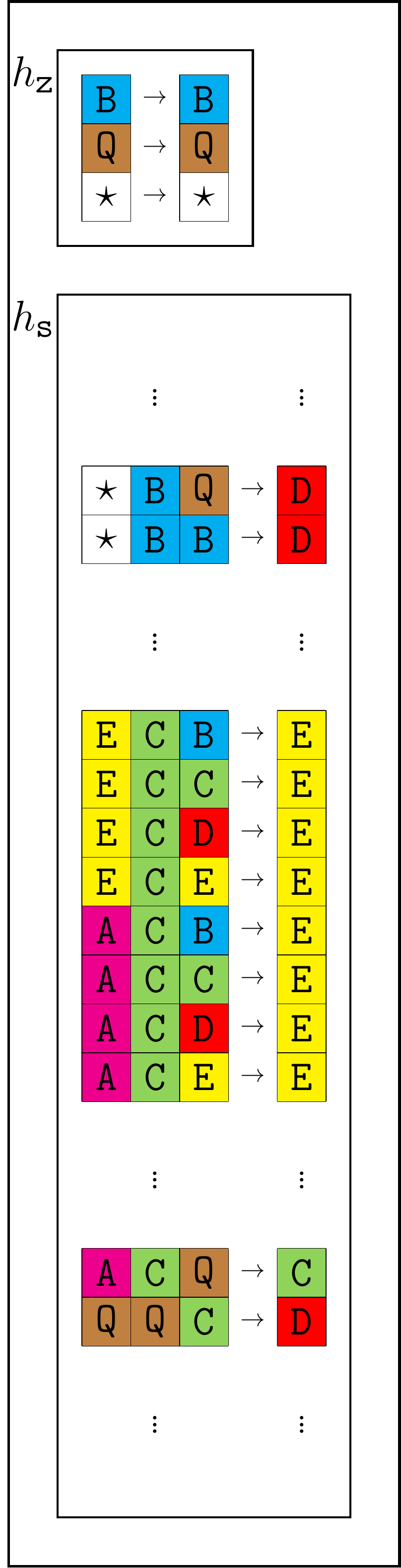}}
\hspace{.2em}
\includegraphics[width=.4\textwidth]{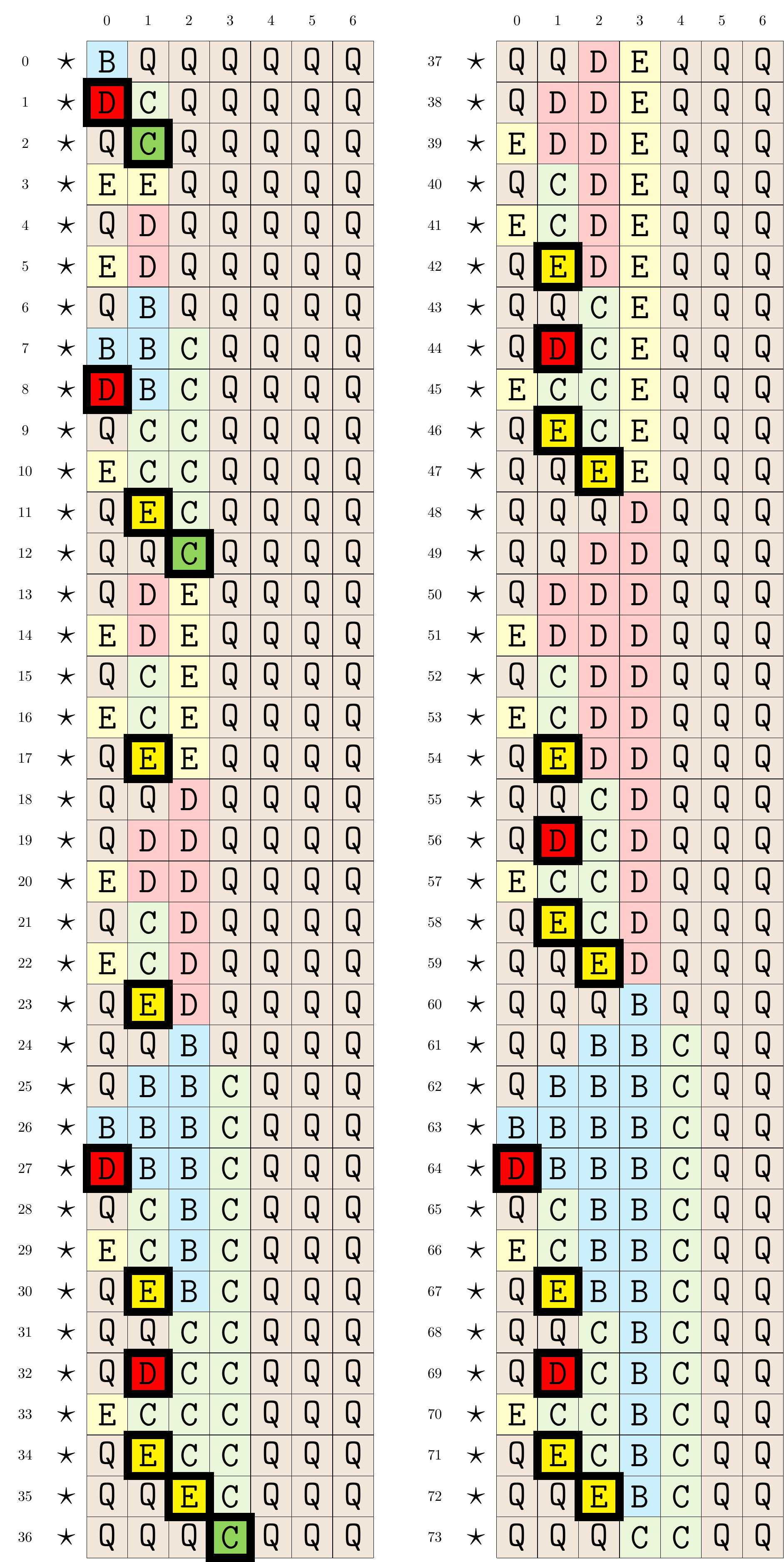}
\caption{6-state diagram, hand-crafted local mapping and resulting 5-state diagram.}
\label{fig:rtsg_dgm}
\end{figure}

The first local mapping that we consider the \emph{identity local mapping} $id$ given by the local transition function of the 6-state solution of Proposition~\ref{prop:umeo-6-state}.
This local mapping simply transforms this solution into itself.
The point here is that we can now work with the local mapping.
However, the reader should be careful to clearly distinguish modifications made the $id$ local mapping, and the resulting modifications in the transition table.
It is easier to think in terms of space-time diagram, since each modification in the local mapping corresponds directly to a uniform set of modifications in the space-time diagram which may or may not be deterministic after these modifications.

Let us now describe how the second, hand-crafted, local mapping is obtained, as we come back on the process itself later.
It is build by noticing different features of the original space-time diagram on the left of Figure~\ref{fig:rtsg_dgm}.
The first thing is that the state $A$ is not often used, so we can try to remove it entirely.
This means changing every entry $(x,y,z)$ of $id_\mathtt{s}$ such that $id_\mathtt{s}(x,y,z) = A$.
But since $A$ is the special generating state, we can not replace it by $B$, $Q$ or $E$ since they already appear in the evolution of the leftmost cell.
So we can consider either $C$ or $D$.
However, looking at time $1$, we see changing $A$ into $C$ would lead to a $CCQ$ local configuration, which is already used.
So we heuristically choose $D$ instead, to have $DCQ$ at time 1, an unused local configuration.
To summarize, for the leftmost cell we choose to change $A$ by $D$, and for the other cells, we can choose any state \emph{a priori}.

The second local mapping is thus obtained by taking every local configurations $(x,y,z)$ of $id_\mathtt{s}$ such that $id_\mathtt{s}(x,y,z) = A$, and setting them to $D$ if $x = \star$, and to $E$ otherwise.
The result is not a deterministic space-time diagram, but this is easily corrected with two additional modifications for $ACQ$ and $QQC$, leading to the local mapping depicted in the center of Figure~\ref{fig:rtsg_dgm}.
The space-time diagram on the right is obtained by applying the local mapping on the space-time on the left as indicated by the outlined local configuration on the left, and resulting state on the right, at the following timestep in direct application of Definition~\ref{def:local-mapping}.

\begin{proposition}
    There is a $n^3$-RTSG solution using 5 states and 72 transitions.
\end{proposition}
\begin{proof}
    First note that the right space-time diagram is deterministic.
    We can therefore extract the transition table given in Figure~\ref{fig:handcrafted} from it.
    No additional transitions appear after the space-time shown in Figure~\ref{fig:rtsg_dgm}.
    To prove this CA to be an $n^3$-RTSG solution, it is enough to check that the local mapping is RTSG-compliant, and since the source CA is an $n^3$-RTSG solution, we can conclude using Proposition~\ref{prop:conservation}.
\end{proof}

\begin{figure}[t]
    \begin{center}
        \includegraphics[width=\linewidth]{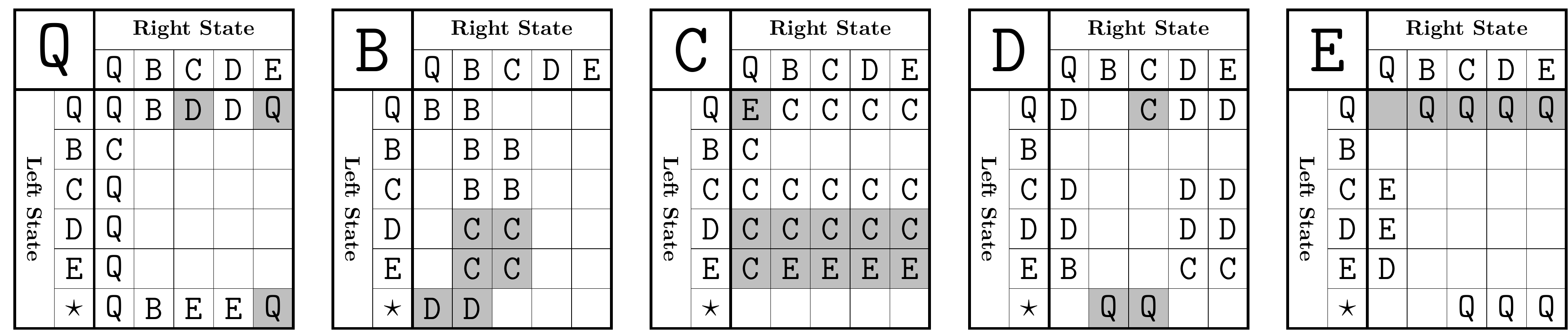}
    \end{center}
    \caption{Transition table of the hand-crafted 5-state solution using 72 transitions.}
    \label{fig:handcrafted}
\end{figure}

To ease the comparison of this 5-state solution with the original 6-state solution, the transitions that are different, added or removed are highlighted in the above table.
Of course, all transition containing $A$ should be considered as removed.
The reader can check that these differences do not correspond exactly to those described in the local mapping.

\subsection{Optimizing Through Millions of Solutions}

Now that we have a first 5-state solution, we are ready to generate millions of them.
The generated solutions are essentially the same, but can have fewer states or/and a different number of transitions.
We begin by a brief summary of the algorithm (more details in~\cite{DBLP:conf/automata/NguyenM20}) and then examine the results.

\subsubsection{The Exploration Algorithm}

The exploration algorithm is related to the hand-crafted process above.
In our case, we make it begin with an identity local mapping of the 5-state hand-crafted solution and let it explore many of its possible modifications by applying one modification at a time.
Only compliant modifications are considered.
In fact, this is a graph exploration algorithm, the node of this graph being compliant local mappings from the hand-crafted 5-state solution to a fixed set of 5 states (and an outside $\star$ state).
The neighbors of a local mapping $\lmap$ are all the local mapping obtained by exactly one compliant modification on $\lmap$.
The identity local mapping is obviously a compliant local simulation, and the algorithm generates all its neighbors and add to the ``remaining tasks'' queue any neighbor that is also a compliant simulation.
Continuing in this way with the content of the queue, the algorithm explores the complete connected component of compliant simulations.
By Proposition~\ref{prop:conservation}, all these compliant local simulations are $n^3$-RTSG solutions.
Let us describe two additional ingredients.

The first one is that there is an initialization step.
To check that a local mapping $\lmap$ is a local simulation, we need generate its local transition relation $\transTab_{\Phi_\lmap}$ to check if it is a function or not.
This is easy to do for all the local mappings if we first collect all the \emph{super-local transition} of the hand-crafted solution, \ie all quintuplets of states with their resulting triplet of states appearing anywhere in the space-time diagram of the hand-crafted solution.
From these data, and for any local mapping $\lmap$, it is enough to apply $\lmap_\mathtt{s}$ on all the super local transitions to generate the entries of associated local transition relation $\transTab_{\Phi_\lmap}$.

The second ingredient is that a parameter $k$ to allow the discovery of more compliant simulation connected component.
Indeed, with $k=0$, the algorithm is unchanged and a compliant local simulation is reached only if its modifications can be applied one at a time while leading to compliant local simulation all the way through.
With $k \ge 1$, the algorithm randomly apply $k$ additional modifications simultaneously on any given compliant local simulation.
If fact, it is often the case that many modifications need to be applied simultaneously, for example the two last modifications described in the design of the hand-crafted solution.
So there is clearly room for improvement in the algorithm in this regard.

\subsubsection{Generated Solutions and Optimizations}

\begin{figure}[t]
    \centering
    \includegraphics[width=.4\textwidth]{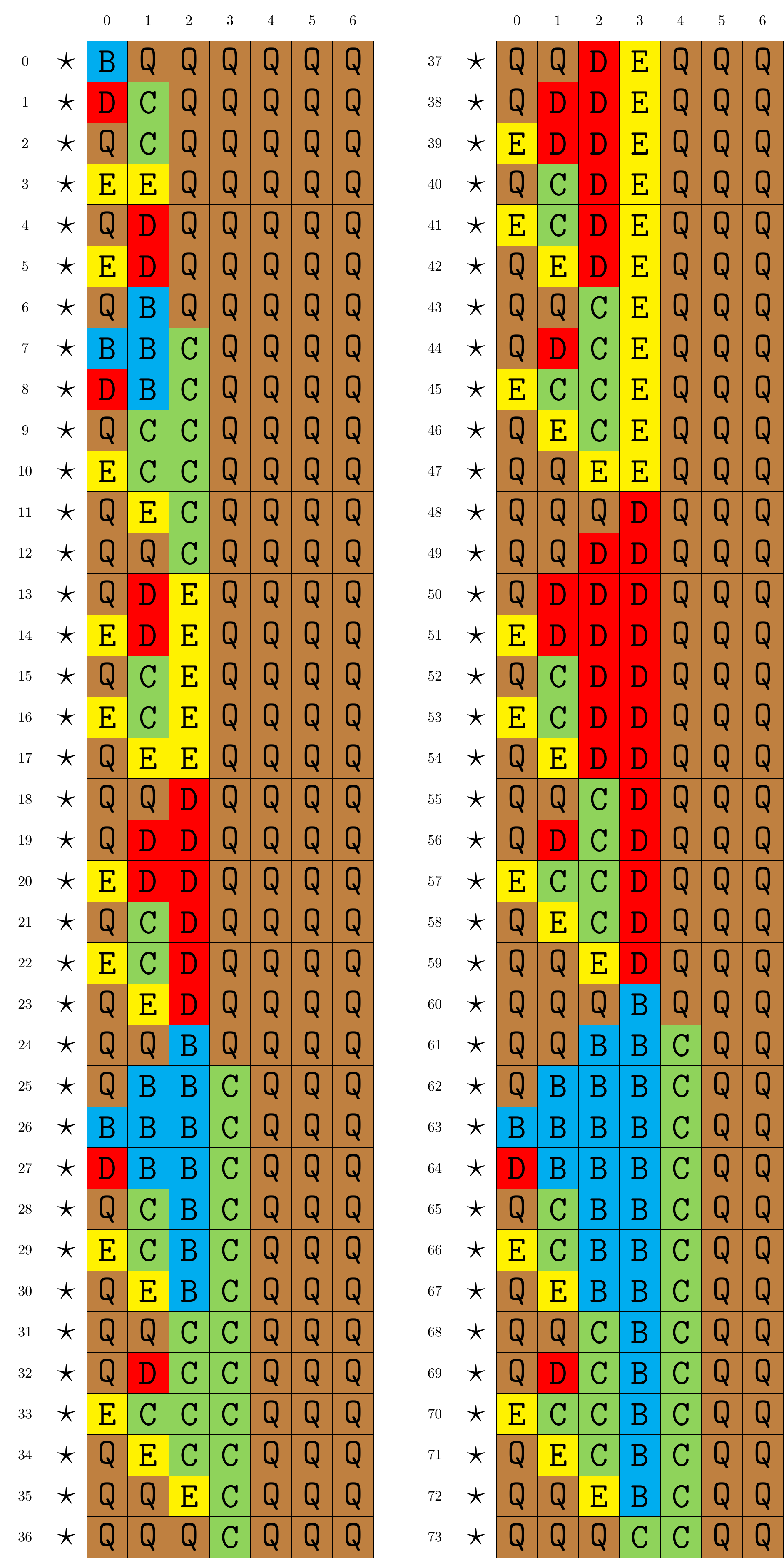}
    \hspace{5em}
    \includegraphics[width=.4\textwidth]{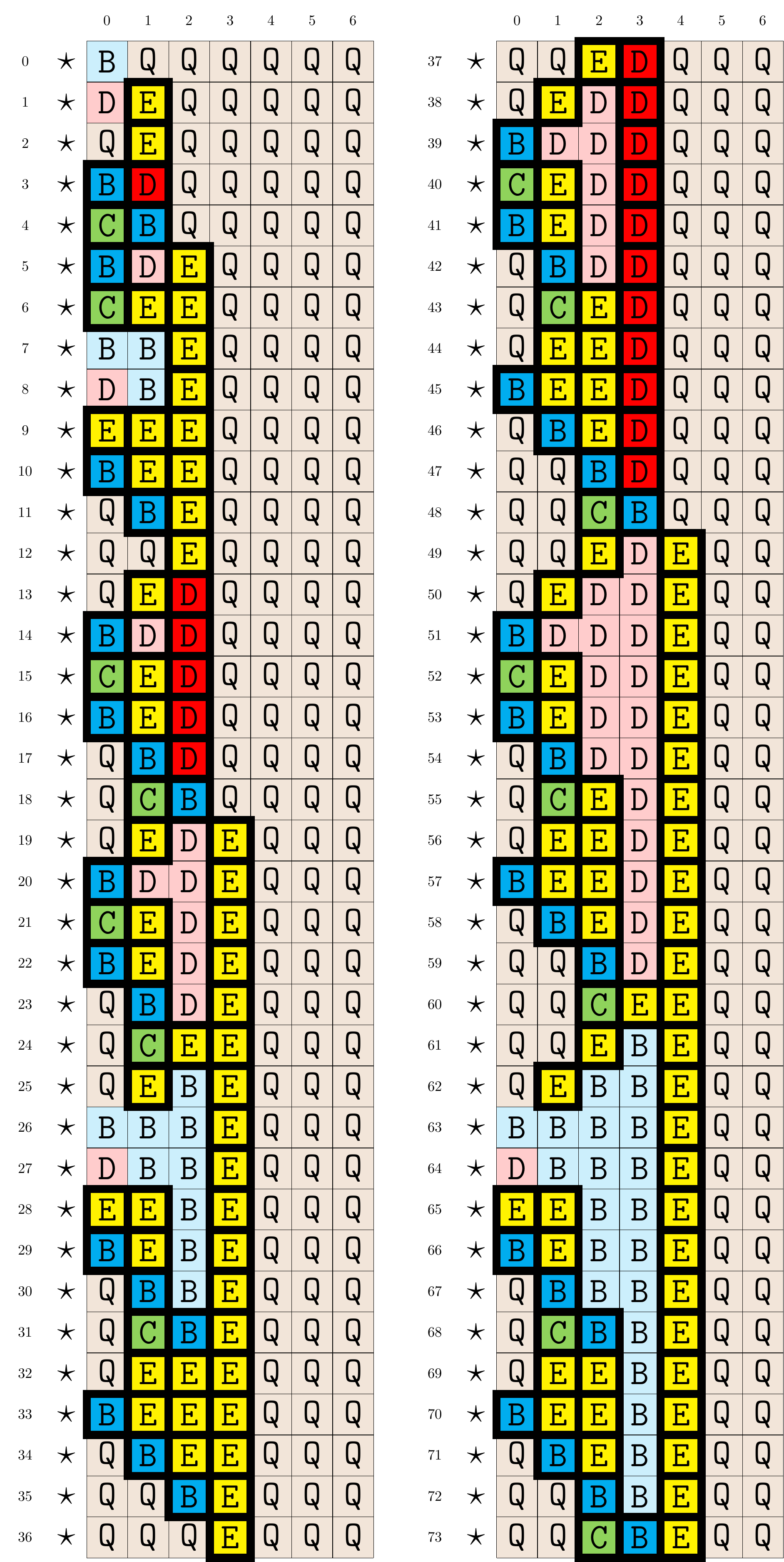}
    \caption{Hand-crafted 5-state diagram and optimized 5-state diagram using 58 transitions.}
\end{figure}

Running the algorithm on a 32 cores of 2.00GHz machine having 126Gb of memory, we obtain so many solutions that the algorithm stops because it runs out of memory resource.
The first time, we ran the algorithm with $k = 0$.
The program actually uses 2 cores and about 43 Gb of memory.
We did not optimize the program nor did we check the configuration of the Java Virtual Machine for this Java implementation.
Since the machine is shared, the following data are not really reproducible, but gives an idea of the execution.
\begin{itemize}
    \item after 1 days, about 15 millions local simulations.
    \item after 6 days, about 85 millions local simulations.
    \item after 20 days, about 90 millions local simulations.
\end{itemize}
The number of solutions found each day was steady for the 6 firsts days then dropped, presumably because of memory issues.
Running concurrently the program with $k = 2$, it uses 2 cores and 36 Gb of memory before it stops because of the same lake of memory.
\begin{itemize}
    \item after 1 days, about 15 millions local simulations.
    \item after 6 days, about 70 millions local simulations.
    \item after 20 days, about 74 millions local simulations.
\end{itemize}
In fact, we had to keep in memory all the solutions and check whether we obtain new solutions up to permutations, in order to be able to have a total number of generated solutions.
Better strategies can be found if the goal is only to optimize the solution.

\begin{proposition}
   There are at least 90,000,000 $n^3$-RTSG solutions using  5 states.
\end{proposition}

Among these millions of solutions, no 4-state solutions are found, but 32379 of them have fewer transitions.
In the following table, the first line indicates a number of transition and the second line the number of solutions having this number of transitions.

\begin{table}[h]
\begin{tabular}{c|c|c|c|c|c|c|c|c|c|c|c|c|c}
58 & 59 & 60 & 61 & 62 & 63 & 64 & 65 & 66 & 67 & 68 & 69 & 70 & 71 \\
\hline
1 & 7 & 22 & 51 & 98 & 174 & 336 & 589 & 1044 & 1618 & 2696 & 4643 & 7671 & 13429
\end{tabular}
\end{table}

\begin{proposition}
   There is a $n^3$-RTSG solution using 5 states and 58 transitions.
\end{proposition}
\begin{proof}
    The transition table of the generated solution is shown in Figure~\ref{fig:optimized}.
    The local mapping having 45 entries different from the identity local mapping, it is not very practical to display it as the previous one, but it is possible to reconstruct it from both cellular automaton.
    It is then a matter of checking that it is compliant and apply Proposition~\ref{prop:conservation} to conclude as before.
\end{proof}

\begin{figure}
    \centering
    \includegraphics[width=\textwidth]{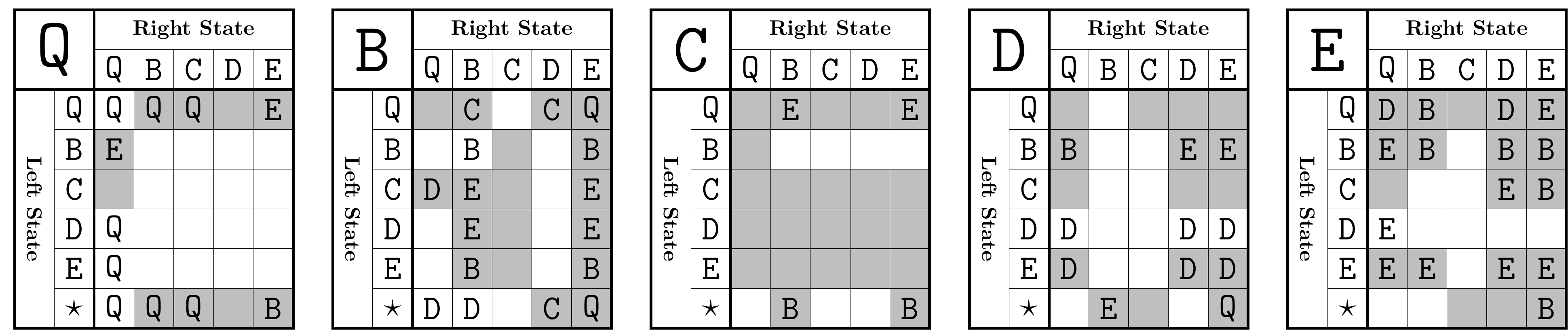}
    \caption{Transition table of the generated 5-state solution using 58 transitions.}
    \label{fig:optimized}
\end{figure}

\subsection{Beyond RTSG and FSSP Optimizations}

It should be clear by now that the approach can be applied to a large class of problems.
For example, the same algorithm used here for the $n^3$-RTSG problem is not particular to the $n^3$ sequence and can be used for any sequence $S$, as clearly indicated in the definitions and propositions above.
Also, the slightly differently parameterized algorithm for the minimal-time FSSP is not particular to minimal-time solutions and can be used for any synchronization time, without even specifying this synchronization time to the algorithm.
The difference in the parameter only reflects the slightly different notion of compliance for RTSG problems and FSSP.
Because the notion of compliance is the only changing factor, the approach can readily be adapted to any class of problem for which an appropriate notion of compliance can be designed.
As examplified here, and in the FSSP case described in~\cite{DBLP:conf/automata/NguyenM20}, the compliance property is a direct translate of the problem.

\section{Conclusion}
\label{conclusion}

There are still many components of this work to communicate properly, including how local mappings compose and relate to each other and how the integration of non-deterministic family of space-time diagrams can allow to explore even more (deterministic) solutions.
Beginning these discussions in this conclusion is not necessarily useful.
There is nonetheless one aspect on which we should comment.
The notion of local mapping appears to be a bridge between a common practice and a topological tool.
Indeed, on the practical side, it is common to work directly at the level of space-time diagrams, and this is practice that is captured formally, and only partly, by local mappings.
This allows to automate this practice.
On the other hand, a question was raised about the relation with conjugacy classes, a standard notion in the cellular automata and symbolic dynamics literature~\cite{DBLP:journals/iandc/JalonenK20}.
In fact, the concept of local mapping appears to be an adaptation of the notion of shift-equivariant homomorphism between two cellular automaton.
Such homomorphisms are usually described on total transition functions, with any configuration being a valid initial configuration.
This is a dynamical system point of view not necessarily aligned with the more algorithmic point of view of FSSP and RTSG problems.
Local mappings augment the notion of homomorphism by including the partiality of the transition functions and the temporal aspect of the space-time diagrams, essential for the very specification of many algorithmic problem.
Forming a bridge between the algorithmic and dynamical points of view might be the reason of their effectiveness.

\bibliographystyle{plainurl}
\bibliography{reference}

\end{document}